\documentclass[letterpaper,twocolumn,10pt]{article}
\usepackage{usenix,epsfig,endnotes}
\def\BibTeX{{\rm B\kern-.05em{\sc i\kern-.025em b}\kern-.08emT\kern-.1667em\lower.7ex\hbox{E}\kern-.125emX}}

\usepackage[utf8]{inputenc}
\usepackage{todonotes}
\usepackage{multicol}
\usepackage{url}
\usepackage{caption,subcaption}  
\usepackage{booktabs}
\usepackage{pgfplots}
\pgfplotsset{width=10cm,compat=1.9}
\usepackage{adjustbox}
\usepackage{mathtools}

\usepackage{amsthm}
\usepackage{breakurl}

  \usepackage[nocompress]{cite}
  \usepackage{cite}
\usepackage[small,compact]{titlesec}
\usepackage{wrapfig}
\usepackage[font=small,labelfont=bf]{caption} 

\usepackage{float}
\usepackage{epsfig}
\usepackage[linesnumbered,algoruled,boxed,lined]{algorithm2e}
\makeatletter
\newcommand{\removelatexerror}{\let\@latex@error\@gobble}
\makeatother

\makeatletter 
\g@addto@macro{\@algocf@init}{\SetKwInOut{Parameter}{Parameters}} 
\makeatother

\usepackage{comment}

\usepackage{pgfplots}
\usepackage{tikz}
\usepgfplotslibrary{groupplots}
\pgfplotsset{compat=newest}

\usepackage{mathtools}
\allowdisplaybreaks

\usepackage{amsthm}
\usepackage{mathtools}
\usepackage{amsmath}
\usepackage{amsfonts}
\usepackage{amssymb}
\newtheorem{theorem}{Theorem}
\newtheorem{lemma}[theorem]{Lemma}

\newtheorem{definition}{Definition}

\begin{document}
\interfootnotelinepenalty=10000

	\author{
		Mohammad M. Jalalzai,
		Chen Feng
	}

\sloppy

\title{
Optimal Latency for Partial Synchronous BFT
\thanks{

}

}

\title{\Large \bf 
Fast B4B: Fast BFT for Blockchains ("patent pending")
}

\author{
{\rm Mohammad M. Jalalzai}\\
The University of British Columbia
\and
{\rm Chen Feng}\\
The University of British Columbia

\and 
{\rm Victoria Lemieux}\\
The University of British Columbia
}

\maketitle
\begin{abstract}
    Low latency is one of the desired properties for partially synchronous Byzantine consensus protocols. Previous protocols have achieved consensus with just two communication steps either by reducing the upper bound on the number of faults the protocol can tolerate ($f = \frac{n+1}{5}$) or use of trusted hardware like Trusted Execution Environment or TEEs. In this paper, we propose a protocol called VBFT, in which the protocol achieves consensus in just two communication steps. VBFT can tolerate maximum number of faults a partial BFT consensus can tolerate ($f = \frac{n-1}{3}$).  Furthermore, VBFT does not require the use of any trusted hardware. The trade-off for this achievement is that at most $f$ nodes  nodes may revert their blocks for small number times. We show that this reversion of a block will not compromise the safety of the protocol at all, yet it may incur a  small amount of additional latency during view change. 
   
\end{abstract}
\maketitle
\pagestyle{plain}

\section{Introduction}\label{Introduction}

Consensus has been a long-studied problem in distributed systems \cite{Consensus-First-Paper,lamport1998thePaxos}. Consensus protocols may either tolerate crash faults \cite{lamport1998thePaxos, lamport2005generalizedPaxos, EPaxos, EpaxosRevisted} or Byzantine faults (arbitrary faults) \cite{Castro:1999:PBF:296806.296824,SBFT,Jalal-Window,Proteus1,HotStuff,ByzantineTolerantGeneralizedPaxos,Aardvark,Prime, avarikioti2021fnfbft,Zeno}. There has been proposals on fast crash failure tolerant protocols that can achieve consensus in just one round trip (RTT) or two communication steps \cite{Atlas,EPaxos,lamport2005generalizedPaxos, EpaxosRevisted} during fast path when certain conditions meet (e.g., no dependency among transactions). 

On the other hand, Byzantine faults are arbitrary faults that include not only crash failure but also more severe faults including software bugs, malicious attacks, and collusion among malicious processors (nodes), etc. Byzantine fault-tolerant protocols achieve fault tolerance by replicating a state, which is often referred to as State Machine Replication (SMR) in the literature. BFT-based protocols have been used in intrusion-tolerant services \cite{Fault-Tolerance-System, Upright-Services, BFT-Critical-Infrastructure,BFT-Scada} including databases \cite{BFT-DataBases}. Recently, BFT-based consensus protocols have been actively used in blockchain technology \cite{Monoxide,SBFT,Jalal-Window,HotStuff,jalalzai2021fasthotstuff,Proteus1,Algorand,Casper}. In all of these use cases, the underlying BFT protocol must remain efficient. More specifically, the protocol response time should be fast. This helps the client requests to get executed and responses returned to the client faster, thereby greatly improving the client experience.

In BFT-based consensus protocols, one of the most important aspects that affect the protocol latency is the number of communication steps. This is even more important when the protocol operates in a WAN (Wide-Area-Network) environment, as the latency for each communication step might be several hundred times higher than the LAN (Local-Area-Network) environment.

Generally, BFT-based consensus protocols in the partially synchronous mode\footnote{There is an unknown maximum bound on the message delay} either operate in three communication steps \cite{Castro:1999:PBF:296806.296824,BFT-SMART,Aardvark} or more \cite{HotStuff,jalalzai2020hermes,jalalzai2021fasthotstuff,SBFT} during normal execution or in the absence of failure. Here,  the normal execution latency is an important metric because
the worst-case latency of partially synchronous BFT consensus protocols can be unbounded. Moreover, failure
(which leads to view changes) is often a rare event in practice. Hence, 
 BFT-based consensus protocols that operate with just two communication steps are regarded as fast BFT consensus protocols \cite{OptimisticBFT3f, FastBFT,Revisiting-Optimal-Resilience-of-Fast-Byzantine-Consensus}.

The proposals in \cite{OptimisticBFT3f} can reach consensus with two communication rounds if all the nodes $3f+1$ ($f$ is the upper bound of the number of Byzantine nodes) in the network are honest. Hence, these protocols lack resilience against Byzantine faults.
Fast-BFT \cite{FastBFT} is another protocol that can achieve consensus in two communication steps, but the maximum number of Byzantine faults it can tolerate during worst case is small $f = \frac{n-1}{5}$. Another recent proposal \cite{GoodCase-Latency-PODC} has further improved the results of \cite{FastBFT} and shown that Byzantine consensus can be achieved in two communication steps where maximum faults tolerated is $f = \frac{n+1}{5}$. 
A concurrent but similar work to \cite{GoodCase-Latency-PODC} is presented in  \cite{Revisiting-Optimal-Resilience-of-Fast-Byzantine-Consensus}. The authors of this work revisit optimal resilience for fast Byzantine consensus and present a tight lower bound for the resiliency $n = 5f-1$ (where $f = \frac{n+1}{5}$)  for fast Byzantine protocols. 
MinBFT \cite{MinBFT} uses trusted hardware TEEs (Trusted Execution Environment) to prevent equivocation attack \footnote{An attack during which a malicious primary/leader proposes multiple blocks for the same height.}. 
Although TEEs are considered to be secure, some recent work (e.g., \cite{TEE-Vulnerability}) raises legitimate concerns about the security of TEEs by identifying several vulnerabilities.

In this paper, we propose a fast BFT-based consensus protocol called VBFT that can achieve consensus during normal execution  with $n=3f+1$ nodes. Hence, it tolerates the \emph{maximum} number of Byzantine faults ($f = \frac{n-1}{3}$), thereby providing the best possible resilience achieved by a partially synchronous  Byzantine protocol \cite{Fischer:1985:IDC:3149.214121}. Therefore, the  resiliency  achieved in this protocol is optimal.
Moreover, it does not require the use of any trusted hardware e.g., TEEs. Like any optimization, there is a trade-off for this improvement. In short, the trade-off is that 
there may be a small latency incurring during the view change process (two communication step)\footnote{ Though this trade-off can be addressed by piggybacking a constant size message during view change.}.

Moreover, there is a possibility that at most $f$ number of honest nodes may revoke their committed block for limited number of times at worst case (during the protocol lifetime) until at most $f$ Byzantine nodes are blacklisted. It should be noted that these revoked transactions will not be considered committed by the clients. Hence, there won't be any double-spending \footnote{A request is considered committed by a client, which is then later revoked.}.
In other words, revocation can only delay the commitment of the transaction for some time. The revocation is only possible if the primary node performs equivocation (propose multiple requests for the same sequence). Since equivocation attacks can easily be detected, hence, the culprit primary can be blacklisted. Therefore, after blacklisting $f$ Byzantine primaries, there will be no Byzantine node left. Hence, in the future, no equivocation will take place. 

We present the VBFT protocol in this paper in the context of blockchains. But  VBFT consensus protocol can simply be adapted for other use cases too.
 VBFT protocol exhibits the following properties:

\begin{itemize}
    \item Optimal latency during normal protocol execution.
    \item Practical throughput, comparable to the state-of-the-art consensus protocol.
    \item Simple by design.
\end{itemize}


\section {Background}
 In PBFT \cite{Castro:1999:PBF:296806.296824} consensus protocol, initially, the primary node broadcasts the proposal. Then there are additional two rounds of the broadcast before an agreement is reached (overall three communication steps). 
  The first round of broadcast is to agree on a request order. Once nodes agree on the sequence for the request, then the request will be assigned the same sequence even if the primary is replaced (view change). Whereas the second round uses broadcast
to agree on committing or executing a request. The two-round broadcast is used to provide a guarantee that if at least a single node commits a request then after recovering from a view change all other nodes will commit the same request. This is due to the reason that if a single node completes the second round of broadcast (hence commits the request), then at least the other two-third of nodes have completed the first round of broadcast. Hence after the view change, these two third nodes make sure the request that has been committed by at least by a single node will be proposed again in the same sequence by the next primary.
This means the first round of broadcast is only necessary for the presence of a failure that will result in the view change. Therefore, this additional round during normal protocol operation can be removed. But if the first round is removed, this means agreement on the requested order and commit occurs in a single round. This is fine during normal protocol execution. But during view change, if at most $f$ honest nodes have committed a request $b$, there is no guarantee that the request will be committed in the same sequence when the new primary is in charge. To avoid this, we introduce a recovery phase at the end of the view change to recover a request $b$ that has been committed by at most $f$ nodes. In this way the new primary will be able to propose the next request, extending $b$. 

But there is also the problem of equivocation by a Byzantine primary.  If the Byzantine primary proposes two equivocating proposals $b$ and $b'$ and one of them is committed by at most $f$ honest nodes, then during the view change, honest nodes cannot agree on the latest committed request to be extended by the new primary. Hence, the protocol will lose liveness \footnote{The protocol will stall indefinitely.}. If the protocol chooses to extend one the of equivocating requests ($b$ or $b'$) then protocol safety cannot be guaranteed. For example, if request $b$ has been committed by at most $f$ nodes, and request $b'$ is chosen (randomly) to be extended, then at most $f$ honest nodes have to revoke request $b$. 

But the good news is that if a request $b$ is committed by at least $f+1$ honest nodes, then it cannot be revoked through equivocation. 
Therefore, though VBFT allows the equivocation and revocation of a request (if the request is committed by at most $f$ nodes) temporarily before all Byzantine nodes get blacklisted, it avoids double-spending. The client only considers a request to be successfully committed by the network if it receives $2f+1$ responses from distinct nodes, verifying that the request has been committed. In this case, out of $2f+1$ nodes (that respond to the client) at least $f+1$ honest nodes have committed the request. We show that a request revocation is not possible once the client considers a request committed ( since at least $f+1$ honest nodes commit a request). Moreover, to avoid equivocation in the future, VBFT simply black lists equivocating primaries. Therefore, after blacklisting at most $f$ Byzantine nodes, no equivocations will take place. Hence, after that, no request will be revoked anymore if it is committed by at least a single node.

The remainder of this paper is organized as follows:
Section \ref{Section: System Model} presents system model, definitions and preliminaries. Section \ref{Section: Protocol} presents detailed VBFT protocol.  Section \ref{Section: Blacklisting} presents how a Byzantine node can be blacklisted. Section \ref{Section:Proof} provides formal proofs of correction for VBFT. 
Section \ref{Section: Related Work}  present related work and Section \ref{Section: Conclusion} concludes the paper.

\section{Definitions and Model}
\label{Section: System Model}
   VBFT operates under the Byzantine fault model. VBFT can tolerate up to $f$ Byzantine nodes where the total number of nodes in the network is $n$ such that $n=3f+1$. The nodes that follow the protocol are referred to as correct/honest nodes. 
 
   Nodes are not able to break encryption, signatures, and collision-resistant hashes. We assume that all messages exchanged among nodes are signed.

To avoid the FLP impossibility result\cite{Fischer:1985:IDC:3149.214121}, VBFT assumes partial synchrony~\cite{Dwork:1988:CPP:42282.42283} model with a fixed but unknown upper bound on message delay.
   The period during which a block is generated is called an epoch. Each node maintains a timer. If an epoch is not completed within a specific period (called timeout period), then the node will timeout and will trigger the view change (changing the primary) process. The node then doubles the timeout value, to give enough chance for the next primary to drive the consensus. 
    As a state machine replication service, VBFT needs to satisfy the following properties.

   \subsection{Definitions}
   \begin{definition}[Relaxed Safety]
A protocol is R-safe against all Byzantine faults if the following statement holds: in the presence of $f$ Byzantine nodes, if \textcolor{black}{$2f+1$} nodes or $f+1$ correct (honest) nodes commit a block at the sequence (blockchain height) $s$, then no other block will ever be committed at the sequence $s$.
\end{definition}

\begin{definition}[Strong Safety]
A protocol is considered to be S-safe in the presence of \textcolor{black}{$f$} Byzantine nodes, if a single correct node commits a block at the sequence (blockchain height) $s$, then no other block will ever be committed at the sequence $s$.
\end{definition}

  \begin{definition}[Liveness]
A protocol is  alive if it guarantees progress 
in the presence of 
at most  \textcolor{black}{$f$} Byzantine nodes.
\end{definition}
   
   \subsection{Preliminaries}
\noindent \textbf{View and View Number.}
A view determines which node is in charge (primary). A view is associated with a monotonically increasing number called view number. A node uses a deterministic function to map the view number into a node $ID$. In other words, the view number determines the primary for a specific view.
 
 \noindent \textbf{Signature Aggregation.} VBFT uses signature aggregation \cite{Short-Signatures-from-the-Weil-Pairing,Boneh:2003,BDNSignatureScheme} in which signatures from nodes can be aggregated into a 
 a single collective signature of constant size. 
 Upon receipt of the messages ($M_1,M_2,\ldots, M_y$ where $ 2f+1 \leq y \leq n$)   with their respective signatures ($\sigma_1,\sigma_2,\ldots, \sigma_y$)
 the primary then generates an aggregated signature $\sigma \gets AggSign(\{M_i, \sigma_i\}_{i \in N})$. The aggregated signature can be verified by replicas given the messages $M_1,M_2,\ldots, M_y$ where $ 2f+1 \leq y \leq n$, the aggregated signature $\sigma$, and public keys $PK_1,PK_2,\ldots,PK_y$. Signature aggregation has previously been used in BFT-based protocols \cite{Castro:1999:PBF:296806.296824,Lamport:1982:BGP:357172.357176, Jalal-Window}. Any other signature scheme may also be used with VBFT as long as the identity of the message sender is known.

 \noindent \textbf{Block and Blockchain}
 Requests are also called a transaction. Transactions are batched into a block data structure. This optimization improves throughput as the network decides on a block instead of a single transaction. Each block keeps the hash of the previous block as a pointer to the previous block. Hence, building a chain of blocks called the blockchain.
 
 \noindent \textbf{Quorum Certificate (QC) and Aggregated QC.}
 A quorum certificate (QC) is the collection of $2f+1$ messages of a specific type from distinct nodes. For example, 
  the collection of at least $2f+1$ votes from distinct nodes for a block is called $QC$ for the block denoted as $QC_b$. A block is certified when enough votes (at least $2f+1$) are received to build its $QC$ or its $QC$ is received. Similarly, collection of $2f+1$ view messages forms $QC$ for the view change or $QC_v$.  
 An aggregated $QC$ or $AggQC$ is collection of $2f+1$ $QC_v$s. If $b$ is the latest committed block or if $QC_b$ is the $QC$ with highest view in the $AggQC$, then the $QC_b$ is called the highQC. 

\SetKwFor{Upon}{upon}{do}{end}
\SetKwFor{Check}{check always for}
{then}{end}


\setlength{\textfloatsep}{1.5pt}

     \SetKwFor{Upon}{upon}{do}{end}
\SetKwFor{Check}{check always for}
{then}{end}
\SetKwFor{Checkuntil}{check always until}
{then}{end}
\SetKwFor{Continuteuntil}{Continue until}
{then}{end}

\begin{algorithm}
\DontPrintSemicolon 
\label{Algorithm: Utilities for nodei}
\label{Algorithm:Utilities}
\SetKwBlock{Begin}{Begin}{}
\SetAlgoLined
\caption{Utilities for node\textit{i}}
\SetAlgoLined

\SetKwFunction{Function}{SendMsgDown}
\SetKwFunction{FuncGetMsg}{MSG}
\SetKwFunction{VoteForMsg}{VoteMSG}
\SetKwFunction{CreatePrepareMsg}{CreatePrepareMsg}
\SetKwFunction{BQC}{GenerateQC}
\SetKwFunction{PQC}{ThresholdGenerateQC}
\SetKwFunction{Matchingmsgs}{matchingMsg}
\SetKwFunction{CreateAggQC}{CreateAggQC}
\SetKwFunction{CreateNVMsg}{CreateNVMsg}
\SetKwFunction{MatchingQCs}{matchingQC}
\SetKwFunction{safeblock}{BasicSafeProposal}
\SetKwFunction{PipelinedSafeBlock}{SafetyCheck}
\SetKwFunction{aggregateMsg}{AggregateMsg}
\SetKwProg{Fn}{Func}{:}{}

\Fn{\CreatePrepareMsg{type, v,s,h,d,$qc_{nr}$, cmd}}{
b.type $\gets$ type\\
b.v $\gets$ v\\
b.s $\gets$ s\\
b.h $\gets$ h\\
b.parent $\gets$ d\\
$b.QC_{nr} \gets qc_{nr}$\\
b.cmd $\gets$ cmd\\
return \textcolor{blue} {{b}}
}
\textbf{End Function}\\

\Fn{\BQC{$V$}}{
qc.type $\gets$ m.type : $m \in V$\\
qc.viewNumber $\gets$ m.viewNumber : $m \in V$\\
qc.block $\gets$ m.block :  $m \in V$\\
qc.sig $\gets$ AggSign( qc.type, qc.viewNumber,\\
qc.s $\gets$ V.s,\\
qc.block,i, \{m.Sig $|  m \in V$\})\\

return \textcolor{blue} {qc}
}
\textbf{End Function}\\


\Fn{\CreateAggQC{$\eta_{Set}$}}{
aggQC.QCset $\gets$ extract QCs from $\eta_{Set}$\\
aggQC.sig $\gets$
AggSign( curView, \{$qc.block |qc.block \in aggQC.QCset $\}, \{$i|i \in N$\}, \{$m.Sig |  m \in \eta_{Set}$\})\\
return \textcolor{blue} {aggQC}
}
\textbf{End Function}\\

\Fn{\CreateNVMsg{$AggQC$,$NextView$}}{
NewViewMsg.type $\gets$ NEW-VIEW\\
NewViewMsg.AggQC $\gets$  $AggQC$\\
NewViewMsg.v $\gets$
$NextView$\\
return \textcolor{blue} {NewViewMsg}
}
\textbf{End Function}\\

\Fn{\PipelinedSafeBlock{b,qc}}{
\If{$b.viewNumber==qc.viewNumber$ $\land b.s=qc.s+1$}{
return \textcolor{blue}{$b.s > curS$}
}
\If{$b.viewNumber>qc.viewNumber$}{
\If{$\beta \land$ B contains $\beta$'s payload }{
return \textcolor{blue}{\textit{b extends from $highQC.block$}}

\If{$\neg \beta$}{
return \textcolor{blue}{\textit{b extends from $highQC.block$}}
}
}

}

}
\textbf{End Function}
\end{algorithm}

\section{VBFT Protocol} 
\label{Section: Protocol}
The VBFT protocol operates in two modes namely normal and the view change mode. The protocol operates in normal mode where blocks are added into the chain until a failure is encountered. To recover from failure, the VBFT switches to the view change mode. Upon successful completion of the view change mode, the VBFT again begins the normal mode.

\subsection{Normal Mode}

 \begin{algorithm*}
\SetKwFunction{Function}{SendMsgDown}
\SetKwFunction{FuncGetMsg}{MSG}
\SetKwFunction{VoteForMsg}{VoteMSG}
\SetKwFunction{CreatePrepareMsg}{CreatePrepareMsg}
\SetKwFunction{QC}{GenerateQC}
\SetKwFunction{Matchingmsgs}{matchingMsg}
\SetKwFunction{MatchingQCs}{matchingQC}
\SetKwFunction{safeblock}{SafeProposal}
\SetKwFunction{aggregateMsg}{AggregateMsg}
\SetKwFunction{PipelinedSafeBlock}{SafetyCheck}
\SetKwProg{Fn}{Func}{:}{}
\SetKwFor{ForEachin}{foreachin}{}{}
\caption{Normal execution for node $i$}
\label{Algorithm: Normal-Mode}
\ForEachin{curView $\gets$ 1,2,3,... }{ 
\If{i is primary}{
\If{If there is $\beta.s == highQC.s+1$}{
    \If{If commands for $\beta$ are present}{
    $B\gets$ CreatePrepareMsg(\textit{Prepare,$v$,$s$,$h$,$parent$,$\perp$,client’s command (from $\beta$)})\\
}
      \Else{
        Broadcast Request for $\beta$ payload 
        
        \Upon{Receipt of $2f+1$ Negative Response \ msgs}{
  
     B$\gets$ CreatePrepareMsg(\textit{Prepare,$v$,$s$,$h$,$parent$,$qc_{nr}$,client’s command})\\
    }
\Upon{Receipt of payload for $\beta$}{
  
     $B\gets$ CreatePrepareMsg(\textit{Prepare,$v$,$s$,$h$,$parent$,$\perp$,client’s command (from $\beta$)})\\
    }     }
     }\Else{
     \Upon{Receipt of $2f+1$ votes}{
      B$\gets$ CreatePrepareMsg(\textit{Prepare,$v$,$s$,$h$,$parent$,$\perp$,client’s command})\\
      }
     }
     broadcast $B$

}
\If{{i is normal replica}}{
\Upon{Receipt of the Request for $\beta$ payload}{
Send payload to the primary or Negative Response if don't have payload
}

 \text{wait for $B$  from primary(curView)}\\

\If{ SafetyCheck($B, highqc$) }  {

Send vote $\langle v,s,h,parent,i \rangle$ for prepare message to primary(curView)

}

 \Upon{Receipt of $2f+1$ votes for block $B$}  {
 execute new commands in $B$\\
 respond to clients
}

}
\Check{For timeout or proof of maliciousness}{
Broadcast $\textsc{view-change}$ message $V$
}
}
\end{algorithm*}

The normal mode is straightforward. The message pattern during normal mode is given in the Figure \ref{fig:Message-pattern}. The primary node receives requests/transactions from clients. A client $c$ sends its signed transaction $\langle REQUEST,o,t,c \rangle$ to the primary. In a transaction the field $o$ is the operation requested by the client, $t$ is the timestamp (primary uses $t$ to order requests from client $c$).

The primary node builds a block (also called proposal/pre-prepare message). The proposal message $\langle \textsc{Pre-prepare}, v,s,h,parent, QC_{nr},cmd \rangle$ include the view number $v$, block sequence (blockchain height) $s$, block hash $h$, parent block hash $parent$ and commands ($cmd$) to be executed. $QC_{nr}$ is the $QC$ for the negative response. After view change before proposing the first block new primary may want to confirm if there is a block that has been committed by at most $f$ nodes. If there is no such block then the primary collects
$2f+1$ negative responses $nr$ and builds $QC_{nr}$. Therefore, after view changes $QC_{nr}$ field may be used only once. More about this will be discussed in the View Change subsection.

The primary signs the proposal over the tuple $\langle v,s,h,parent \rangle$. The primary node then proposes the proposal to all nodes through broadcast during the pre-prepare phase (Algorithm \ref{Algorithm: Normal-Mode}, lines 18-22). Upon receipt of a proposal, a node $i$ verifies the primary signature and the message format. Each node then makes a safety check through $SafetyCheck()$ predicate. The $SafetyCheck()$ predicate accepts two parameters including the received block $b$ and the $highQC$ (latest $QC$ built from $2f+1$ votes received for the previous block). In the absence of failure, the $SafetyCheck()$ predicate makes sure the proposal extends the latest committed block. Which means if the sequence of the latest committed block $b'$ is $s$ (the block has $highQC$), then the sequence of the received block $b$ has to be $s+1$. If $b.s> b'.s+1$, then it either means the node $i$ is missing one or more committed blocks or the block proposal has invalid sequence. In the case of missing blocks the node will download the missing blocks and their respective $QC$s. In case, of invalid block sequence, node $i$ can send the block $b'$ metadata (with invalid sequence) as a proof to the network to trigger a view change.

If there is a failure or view change, then $SafetyCheck()$ predicate makes sure any uncommitted block seen by at least $f+1$ honest nodes (a block committed by at most $f$ honest nodes) has been re-proposed with its previous sequence number \footnote{More details about this in later section}. If the safety check is successful, then the node $i$ broadcasts its vote $\langle Vote, v,s,h,parent,i \rangle$ signed over tuples $\langle  v,s,h,parent,i \rangle$  (Algorithm \ref{Algorithm: Normal-Mode} lines 29-32). 

Upon receipt of $2f+1$ votes for a block $b$ each node commits the block (commit phase). Each node also builds a $QC$ from votes it received for the block $b$. This $QC$ will be used as a proof of commit during view change or for providing proof of commit required by application that is running over the top of the VBFT protocol. 
Each node that has committed the block will also send a $Reply$ message to clients whose transaction is included in the block (Algorithm \ref{Algorithm: Normal-Mode} lines 32-35). A client considers its request/transaction being committed if it receives $2f+1$ distinct $Reply$ messages.
Similarly, upon receipt of $2f+1$ votes, the primary commit the block and proposes the next block in its queue.

If a client $c$ does not receive $Reply$ message within some specific interval, then the client broadcasts its request to all nodes. Upon receipt of the request, if the request has been committed, then each node sends $Reply$ message. If it has not been proposed by the primary, nodes forward the request to the primary. 
If the primary still does not propose the block within the timeout period, then each node broadcasts a $\textsc{view-change}$ message to replace the primary. Similarly, if a block is not committed within a timeout period, the primary sends an invalid message or performs equivocation, nodes will trigger the view change process. The view change process will result in replacing the faulty primary. 

\begin{figure*}
    \centering
    \includegraphics[width=14cm,height=6cm]{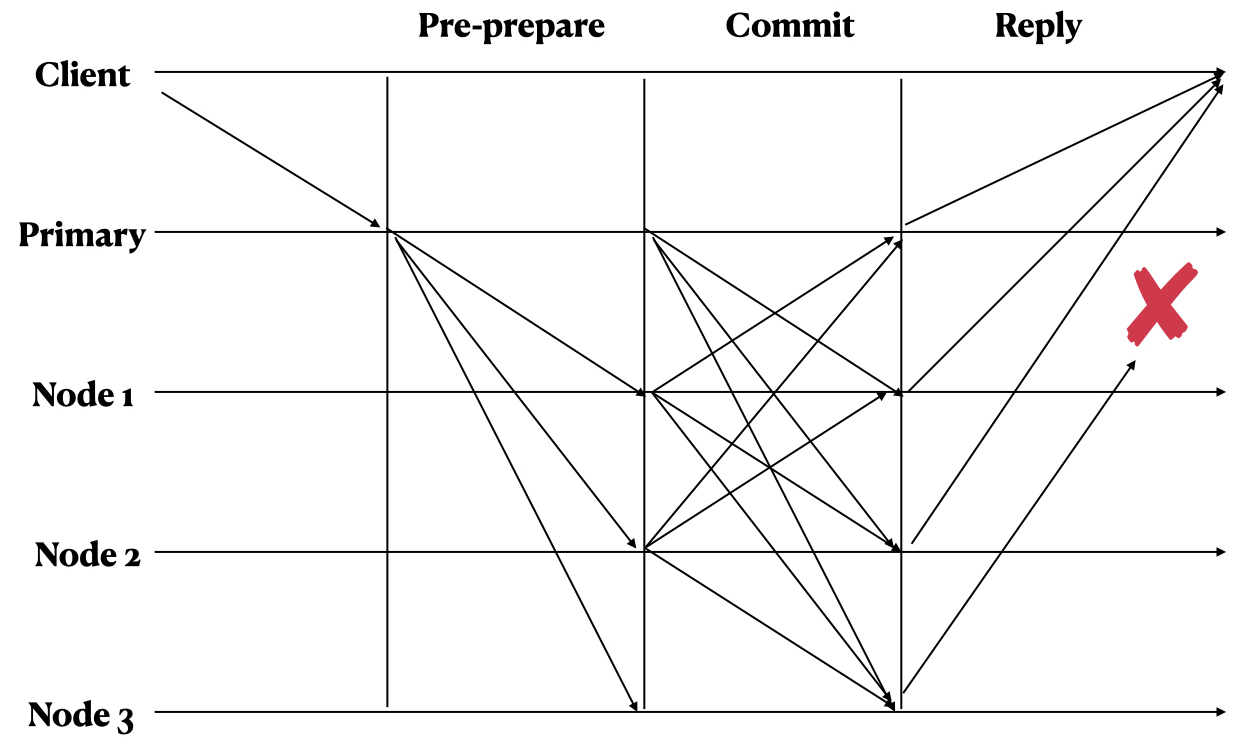}
    \caption{\textbf{VBFT Message Pattern}.}
    \label{fig:Message-pattern}
\end{figure*}

\subsection{View Change}   
    In BFT-based protocols, a view change is used to replace the failed primary with a new primary. Nodes use a deterministic function to map the view number to the primary ID. Primary can be chosen in round-robin manner \cite{HotStuff, Castro:1999:PBF:296806.296824,Jalal-Window} or randomly \cite{Coin-Flipping-by-Telephone,Coin-Flipping-First-Paper,Coin-Flipping-Popov,Random-Oracles-in-Constantinople,An-Optimally-Fair-Coin-Toss,Proteus1,hanke2018dfinity,jalalzai2020hermes}.
    
    In addition to selecting a new primary, view change also involves the synchronization of information among nodes. More specifically, the new primary needs to determine the sequence number that it will use for its first request to be proposed. The block proposed in the new sequence number just after the view change has to extend the latest block committed by at least one honest node \cite{SBFT,HotStuff,Castro:1999:PBF:296806.296824} (S-safety). In VBFT, the view change is different than the view change in the general BFT-based consensus protocol. 
    In VBFT, there are two ways to extend the latest committed block.
    If the latest proposed block has been committed by more than $f$ honest nodes, then the VBFT view change is similar to the view change in classic BFT consensus. But if the latest proposed block is committed by at most $f$ honest nodes or has not been committed by any node then the VBFT, requires an additional data retrieval phase (two communication steps) to address this case. Moreover, if at most $f$ honest nodes have committed a block, then the equivocation by a Byzantine primary may result in revocation of the latest committed block during view change. We, therefore, show that after each revocation the Byzantine primary will be blacklisted. Hence, after blacklisting $f$ nodes the protocol will become S-safe. Which means VBFT will guarantee that if a block is committed by one node, then eventually it will be committed by all nodes. Furthermore, as stated, this revocation does not affect clients and will not result in double-spending.

 \begin{algorithm}[ht]
\SetKwFunction{Function}{SendMsgDown}
\SetKwFunction{FuncGetMsg}{MSG}
\SetKwFunction{VoteForMsg}{VoteMSG}
\SetKwFunction{CreatePrepareMsg}{CreatePrepareMsg}
\SetKwFunction{QC}{GenerateQC}
\SetKwFunction{Matchingmsgs}{matchingMsg}
\SetKwFunction{MatchingQCs}{matchingQC}
\SetKwFunction{safeblock}{SafeProposal}
\SetKwFunction{aggregateMsg}{AggregateMsg}
\SetKwProg{Fn}{Func}{:}{}


\SetKwFor{ForEachin}{foreachin}{}{}
\caption{View Change for node $i$}
\label{Algorithm: View-Change}

\If{i is primary}{
\If{ $n-f$ $V$ (view change) msgs are received}{
 $aggQC \gets    CreateAggQC(V_{Set})$\\
 $NvMsg \gets$ CreateNVMsg(aggQC,v+1)
Broadcast $NvMsg$

}

\If{ $n-f$ $R$ (view change) msgs are received}{
 $qc_r \gets    CreateReadyMsg(R_{set})$
 
Broadcast $qc_r$
}

}
\If{i is normal node}{

\Upon{Receipt of NEW-VIEW Msg}{

Extract the latest $QC$ and $\beta$ (if present) 

Send Ready Msg $R$ to the primary

}

\Upon{Receipt of $qc_r$}{

Execute Algorithm \ref{Algorithm: Normal-Mode}

}

}
\end{algorithm}

   If a node does not receive a message from primary or a command is not executed within some interval, then the node broadcasts a $\textsc{view-change}$ message $V=\langle  \textsc{View-Change},v+1, qc, \beta, \perp, i\rangle$.  The $qc$ is the $QC$ for the latest committed block. Whereas, the $\beta$ is the header of the latest uncommitted block for which the node $i$ has voted.
   In case the node receives a proof of maliciousness from primary e.g, multiple proposals for the same sequence, invalid message etc, then it can include the proof too in the $\textsc{view-change}$ message ($V=\langle  \textsc{View-Change},v+1, qc, \beta, proof, i \rangle$) (Algorithm \ref{Algorithm: Normal-Mode} lines 37-39). Once the new primary receives $2f+1$ \textsc{View-Change} messages (may include its own), it prepares \textsc{new-view} message $\langle \textsc{new-view}, AggQC, v+1 \rangle$. Here $AggQC$ is the QC built from $2f+1$ \textsc{View-change} messages. The new primary then broadcasts $\textsc{new-view}$ message to all nodes (Algorithm \ref{Algorithm: View-Change} lines 2-5).
   Each node that receives  the $\textsc{new-view}$ can extract important information out of it  as described below:
   \begin{enumerate}
       \item The $QC$ type field in the $\textsc{view-change}$ message in $2f+1$ $QC$s in $AggQC$, is used to extract the latest committed block by at least $f+1$ honest nodes. Since if $f+1$ honest nodes commit the latest block, then any combination of $2f+1$, $\textsc{view-change}$ messages will have at least one $QC$ for the committed block by at least $f+1$ nodes (for more details please refer to Lemmas \ref{Lemma:R-safe-Normal} and \ref{Lemma: R-safety-ViewChange}). 
       Though the $QC$ of a block committed by at most $f$ nodes may also end up in $AggQC$, it is not guaranteed. If it did, then the new primary will extend the block committed by at most $f$ nodes.

        \item The $\beta$ field in the $2f+1$ $\textsc{view-change}$ messages in $AggQC$
        determine the latest committed block by at most $f$ honest nodes (In the absence of equivocation). If a block is  committed by less than $f+1$ honest nodes, then its $QC$ may not end up among $2f+1$ $\textsc{view-change}$ messages in $AggQC$, but since $2f+1$ nodes (of which $f+1$ are honest) have voted for that block therefore, at least one $\beta$ in $AggQC$ will be from the latest committed block. There is also the possibility that the block/request of $\beta$ may have not been committed by a single honest node. In that case, the next primary will have to collect the proof that the block from $\beta$ has not been committed.
   \end{enumerate}

      Therefore, upon receipt of $\textsc{new-view}$ message from the new primary, each node extracts the latest $qc$ as well as $\beta$ (if there is any) as shown in Algorithm \ref{Algorithm: View-Change}, lines 12-15.
      Each node then sends $Ready$ message to the new primary. The new primary aggregates $Read$ messages into the $QC$ for the ready message ($QC_{r}$) and broadcasts back. Upon receipt of $QC_{r}$ each node is now ready to take part in the next view.
       
       \textbf{Recovering the block committed by at most $f$ honest nodes using $\beta$.} 
       As stated $\beta$ in the view change $QC$ ($QC_v$) is used to recover any block that has been committed by at most $f$ nodes (in the absence of equivocation). But if there is not such a block, then the primary has to provide proof that there is no block such that it has been committed by at most $f$ honest nodes.
       
       If $\beta$ field in all $\textsc{view-change}$ messages in $AggQC$ is nil, then it means no block has been committed by at most $f$ honest nodes. In other words, the latest proposed block has been committed by at least $f+1$ honest nodes. Therefore, the new primary will extend the block of the latest $QC$ in the  $AggQC$.  In case there is a block $b$ that has been committed by $x$ number of nodes such that $f \geq x \geq 1$, then the $QC$ for block $b$ may or may not end up in $AggQC$. In this case 
       the protocol relies on $\beta$ (for block $b$) to make sure the latest block gets extended after the view change.
       The $\beta$ is only relevant if its sequence is greater than the sequence of latest $QC$  ($highQC$) in $AggQC$ such that $\beta.s == highQC.s+1$.
        If at least one $\textsc{view-change}$ message in $AggQC$ has the $QC$ for the latest committed block $b$ ($highQC$), then the new primary will propose a block that will extend the block $b$. In this case $\beta$ is not relevant since $\beta.s \leq highQC.s$.
        It is also possible that no $QC$  for block $b$ in the $AggQC$ of the $\textsc{new-view}$ message is collected during view change. But it is guaranteed that at least one $\beta$ in $AggQC$ will be from the block $b$. In this case, if the new primary does not have the payload for the $\beta$ from block $b$, 
       then the new primary will be able to retrieve the contents of block $b$ from other nodes in the network. The new primary then re-proposes the same block $b$ at the same sequence again as the first proposal. Hence, in case if there is a non-nil $\beta$, then the new primary has to make sure if there is a block that has been committed (by at most $f$ honest nodes), then it should be recovered and re-proposed. It is also possible that a block may not be committed by any honest node but its $\beta$ has been received by the new primary, hence, included in the $AggQC$.
       In this case, the new primary will try to retrieve the block for $\beta$ (when $\beta.s==highQC.s+1$). If the block is retrieved, the primary will propose it for the same sequence. If the block is not retrieved then the primary has to provide proof to all honest nodes that such a block has not been committed by any honest node. The proof will allow the new primary to extend the block with the latest $QC$ (instead of the block of $\beta$) in the $\textsc{view-change}$ message.
       
       Therefore as it can be seen in the Algorithm \ref{Algorithm: Normal-Mode}, if there is a $\beta$ such that $\beta.s==highQC.s+1$ and the primary has the payload for $\beta$, then it can propose the $\beta$'s payload (lines 4-6). If the $\beta$ payload is not present, then the algorithm requires to recover the payload (block) (lines 7-11). The primary broadcasts a request to retrieve the payload for $\beta$. Now if $\beta$ has been committed by at least one node, then this means $2f+1$ nodes have voted for it. Which means $2f+1$ nodes have the payload for $\beta$. Therefore when nodes (at least $2f+1$) respond to the request for $\beta$ payload, at least one response will have the payload. Upon, receipt of the payload the new primary proposes the payload for $\beta$ again (lines 12-15).
       
       There is also a possibility that the respective block for the $\beta$ may not be committed by any honest node. In that case, the primary may not receive the $\beta$ payload. But it is guaranteed that the new primary will receive $2f+1$ negative responses from distinct nodes. These $2f+1$ negative response indicates that the block for the $\beta$ has not been committed by any honest node. To prove that the respective block for $\beta$ has not been committed, the new primary attaches the $QC$ for negative response ($qc_{nr}$) to its first proposal for the view. If there is a $\beta$ and the new primary neither proposes its payload nor includes its $qc_{nr}$ during the first proposal, honest nodes will trigger view change. Therefore,  in the absence of equivocation, the VBFT guarantees S-safety. 
       Therefore, the presence of the latest block $QC$ or $\beta$ is guaranteed in $AggQC$ (further details in Section \ref{Section:Proof}). Hence, the new primary will extend the latest committed block (in the absence of equivocation).
       
       \textbf{Addressing the problem of Equivocation (Byzantine Primary).}
       Since  VBFT achieves consensus in just two steps of communications, therefore it allows Byzantine primary to perform equivocation (proposing multiple blocks for the same sequence).
       In the case of equivocation if one of the proposed blocks is committed by $f+1$ honest nodes then 
       the equivocation will not result in revocation of the block committed by $f+1$ honest nodes.
       Hence, equivocation will not be able to cause any safety or liveness issues. But 
       the problem arises if at most $f$ honest nodes commit one of the equivocated blocks. In this case, there is a possibility that the block committed by at most $f$ nodes will be revoked. To better explain how revocation takes place, let us assume that there are two equivocated blocks $b_1$ and $b_2$, such that $b_1$ has been committed by $f$ honest node. During the view change, the new primary receives $\textsc{view-change}$ messages from $2f+1$ nodes (assuming that the $QC$ for $b1$ is not among the $2f+1$ $\textsc{view-change}$ messages). 
       The new primary may receive $\textsc{view-change}$ messages  containing
       $\beta$ for both blocks $b_1$ ($\beta_1$) and $b_2$ ($\beta_2$).
       Now if the new primary chooses to recover the payload for $\beta_1$ and repropose the payload for $\beta_1$, then the S-safety does not break. But if the new primary choose to extend $\beta_2$, then up to $f$ nodes that have committed $b_1$, have to revoke $b_1$. Hence, the S-safety does not hold if equivocation is performed. Since at most $f$ honest nodes have committed block $b_1$, therefore, the clients have not received $2f+1$ $Reply$ messages. As a result, clients still do not consider transactions in the block $b_1$ as committed. Therefore, the revocation does not result in double-spending. A Client $c$ can only get $2f+1$ $Reply$ messages for a transaction if its subsequent block is committed by at least $f+1$ honest nodes. As stated, if a block is committed by $f+1$ honest nodes then it will not be revoked (See safety proof for more details). Therefore, to discourage Byzantine nodes equivocated primaries will be blacklisted. As a result, once all Byzantine nodes are blacklisted, there won't be any equivocation.

\textbf{View Change Optimization}
View change in classic BFT protocols \cite{SBFT,Castro:1999:PBF:296806.296824} can have quadratic cost of signature verification in $AggQC$.
Recently there have been proposals on improving the message and computational costs from signatures during the view change \cite{No-Commit-Proofs, jalalzai2021fasthotstuff}. We preferred to use the solution in cite{jalalzai2021fasthotstuff} as it is straightforward and easy to implement.
In Fast-HotStuff \cite{jalalzai2021fasthotstuff}, authors show that this quadratic cost of signature verification can be reduced linear. The technique they suggested is to verify the aggregated signature of the $AggQC$ along with the aggregated signature of the $\textsc{view-change}$ message with the highest $QC$. This is sufficient to guarantee the validity of the $AggQC$.  The verification of the aggregated signature of the $AggQC$ verifies that it contains $\textsc{view-change}$ message from $2f+1$ distinct nodes. Whereas, the verification of the signature of the latest $QC$ makes sure the latest $QC$ and potential latest $\beta$ messages are valid. Thus, each node does not require to verify aggregated signatures from the remaining $2f$ $QC$s. Although this solution does not help with the message complexity during the view change, in practice the collective size of view change messages is negligible compare to the block size. Hence, their effect on performance is also negligible \footnote{If the network size is too large or the block size is too small, then view change message complexity may affect the performance.}.

\section{Black Listing Byzantine Primary}    
\label{Section: Blacklisting}
       An equivocation attack is mainly performed to break the protocol safety. Although it is a strong type of attack, the good news is that it can easily be detected. It is also important that the blacklisting decision is consistent throughout the network.  Therefore, the blacklisting of a Byzantine node is done through a transaction (request). This allows the network to take a decision on blacklisting and remain consistent.
       
       When an honest node $i$ receives multiple blocks for the same sequence, it knows that the primary is Byzantine and is trying to perform equivocation. Therefore, broadcasts the proof (which includes the header of both equivocated blocks)  included in the view change message. This will trigger a view change as explained in the algorithm \ref{Algorithm: View-Change}. After the view change, honest nodes expect the new primary to include the blacklisting transaction in the block proposal. If the new primary does not propose a blacklisting transaction within the timeout period, then honest nodes forward the proof to the new primary and wait again for the blacklisting transaction to be proposed. If by the timeout the blacklisting transaction is not proposed by the primary, then honest nodes trigger view change by broadcasting $\textsc{view-change}$ messages. 
       This process will continue until the new primary proposes the transaction containing a blacklisting transaction.
       
       It is also possible that a node is no longer the primary but another node $i$ discovers proof of equivocation against it. In this case, the node $i$ still broadcasts the proof to all nodes. Upon receipt of the proof a $j$ checks if the node that has performed equivocation has been blacklisted or not. If the equivocator is blacklisted then the proof is ignored by the node $j$. If not then it forwards the proof to the current primary and waits for the proposal of the blacklisting transaction. If the primary does not propose the blacklisting transaction within a timeout interval, then node $j$ broadcast $\textsc{view-change}$ message. Therefore, eventually, there will be an honest primary that will propose the blacklisting transaction. 
       
\section{Proof of Correctness}
\label{Section:Proof}
In this section, we provide proof of safety and liveness for the VBFT protocol. We consider $S_h$ as the set of honest nodes in the network and $2f+1 \leq |S_h| \leq n$.
\subsection{Safety}
As stated VBFT satisfies S-safety during the normal protocol. VBFT also satisfies S-safety during a view change in the absence of equivocation. VBFT protocol may fall back to R-safety. But once all byzantine nodes (at most $f$) get blacklisted then, the protocol will always guarantee S-safety. Below we provide lemmas related to the VBFT safety.
 \begin{lemma}
       \label{Lemma:S-safe-Normal}
        VBFT is S-Safe during normal execution of the protocol.
       \end{lemma}
       \begin{proof}
       This lemma can be proved by contradiction. Let us assume that a block $b_1$ has been committed by a single  node $i$ such that $i \in S_h$ in the sequence $s$. This means $2f+1$ nodes in a set $S_1$ have voted for block $b_1$. 
       Similarly an honest node $j$ ( $j \in S_h$ )  commits a block $b_2$ in the sequence $s$. A set of $2f+1$ nodes ($S_2$) have voted for the block $b_2$. Since $n=3f+1$, we have $S_1 \cap S_2 \geq f+1$. This means there is at least one honest node that has voted both, for the block $b_1$ and $b_2$. But this is impossible since an honest node only vote once for a single block at the same sequence. Hence, VBFT is S-safe during normal execution of the protocol.
       
       \end{proof}
       
       \begin{lemma}
   \label{Lemma:R-safe-Normal}
   VBFT is R-Safe during normal execution of the protocol.
   \end{lemma}
  \begin{proof}
    From Lemma \ref{Lemma:S-safe-Normal}, we know that VBFT is S-safe during normal mode. Therefore, we can conclude that if VBFT is S-safe during the normal mode, then it is also R-safe during the normal mode.
  \end{proof}
  
  \begin{lemma}
     VBFT is R-Safe during view change.
     \label{Lemma: R-safety-ViewChange}
  \end{lemma}

\begin{proof}
  Let us assume that a block $b_1$, is committed just before view change by a set of nodes $S_1 \in S_h$ of size $f+1$. During view change $2f+1$ nodes (out of which at least $f+1$ are honest nodes) in set $S_2$ send their $\textsc{view-change}$ messages to the new primary. The new primary aggregates $2f+1$ $\textsc{view-change}$ messages into $\textsc{new-view}$ message
  as shown in (Algorithm \ref{Algorithm: View-Change}).
  Since $S_1 \cap S_2 = S$ and $|S| \geq 1$, therefore there is at least one honest node that have committed block $b_1$ and has its $\textsc{view-change}$ message $V \in$ $\textsc{new-view}$.    
  This means when the new primary broadcasts the 
  $\textsc{new-view}$ message, every receiving node will know that block $b_1$ is the latest committed block. Therefore, the new primary will have to extend the block $b_1$. 
\end{proof}

\begin{lemma}
    VBFT is S-Safe during view change when there is no equivocation.
         \label{Lemma: S-safety-ViewChange-NoEq}

\end{lemma}
\begin{proof}
  A block $b$ is committed by a single honest node $i$ at the sequence $s$ just before the view change. This means a set of nodes $S_1 \geq 2f+1$ have voted for the block $b$. During the view change the new primary collects $\textsc{view-change}$ messages from another set of $2f+1$ nodes ($S_2$) into $\textsc{new-view}$ message. There is no guarantee that node $i$'s $\textsc{view-change}$ message is included in the $\textsc{new-view}$ message built from $\textsc{view-change}$ message of nodes in $S_2$. If  node $i$'s $QC$ for block $b$ from its $\textsc{view-change}$ message is included in the $\textsc{new-view}$ message, then the new primary will simply propose a block that extends the block $b$. But in case the $QC$ for block $b$ is not included in $\textsc{new-view}$, then block $b$ can be recovered using information in the $\beta$ field of the $\textsc{new-view}$ messages in $\textsc{new-view}$ message.
 Each node in $S_2$ has included the latest $\beta$ (block header it has voted for) in its $\textsc{view-change}$ message it has sent to the primary. The primary has aggregated these messages into $\textsc{new-view}$ message. Since $n=3f+1$ therefore, $S_1 \cap S_2 \geq f+1$. This means there is at least one honest node $j$ in the intersection of $S_1$ (that has voted for the block $b$) and $S_2$ (has its $\textsc{view-change}$ message $V_j \in \textsc{new-view}$ message such that and $V_j.\beta.s=highQC\footnote{$highQC$ is the $QC$ with highest sequence (latest $QC$) in the $AggQC$. $AggQC$ is the set of $QC$s in the $\textsc{new-view}$ message.}.s+1$). 
  Therefore, the new primary can retrieve the payload for $\beta$ (block $b$'s header) and propose it for the same sequence $s$ as shown in Algorithm \ref{Algorithm: Normal-Mode}. 
\end{proof}

\begin{lemma}
    After blacklisting $f$ Byzantine nodes, the VBFT will guarantee S-safety.
     \label{Lemma: S-safety-After-BlackListing}
\end{lemma}

\begin{proof}
A Byzantine primary may perform equivocation by proposing multiple blocks for the same sequence. As stated in Sections \ref{Section: Protocol} and \ref{Section: Blacklisting} an equivocation attack can easily be detected and the culprit node can be blacklisted. Therefore, upon blacklisting $f$ number of Byzantine nodes, there won't be any Byzantine node left to perform equivocation when selected as a primary. Since there will be no equivocation anymore therefore based on the Lemmas  \ref{Lemma:S-safe-Normal}, \ref{Lemma:R-safe-Normal}, \ref{Lemma: R-safety-ViewChange}, and \ref{Lemma: S-safety-ViewChange-NoEq}, VBFT protocol is S-safe after blacklisting $f$ equivocating nodes. 
\end{proof}

The above lemmas prove that the VBFT protocol is always R-safe. Fast-B4B is S-safe during the normal mode and view change (when there is no equivocation). After, blacklisting $f$ Byzantine primaries, due to equivocation the VBFT protocol will always guarantee S-safety.

\subsection{Liveness}
A consensus protocol has guaranteed that it makes progress eventually and will not stall indefinitely.  VBFT replaces a primary through a view change if the primary fails to make progress. VBFT borrows techniques from PBFT \cite{Castro:1999:PBF:296806.296824} to provide liveness. The first technique includes an exponential back-off timer where the timeout period is doubled after each view change to give the next primary enough time to achieve the decision on a request/block 

Secondly,  a node broadcasts $\textsc{view-change}$ message if it receives $f+1$ $\textsc{view-change}$ messages higher than its current view from distinct nodes. This guarantees that at least one message is from an honest node. Moreover, if a node receives $f+1$ $\textsc{view-change}$ messages for different views greater than its current view, it will broadcast a $\textsc{view-change}$ message for the smallest view of the $f+1$ $\textsc{view-change}$ messages.

Third, to prevent the protocol to be indefinitely in the state of view change, a node will not broadcast $\textsc{view-change}$ message and trigger view change if it receives at most $f < n/3$ number of $\textsc{view-change}$ messages. Therefore, Byzantine nodes cannot trigger a view change until at least one honest not also broadcast a $\textsc{view-change}$ message.

If the primary is selected in round robin-manner, then after at most $f$ Byzantine primaries, an honest primary will be selected. Therefore, progress will be made eventually. Similarly, if the primary is being selected randomly then the probability of a bad event (selection of a Byzantine primary) is $P_b=1/3$. By considering such a bad event as the Bernoulli trial, we have the probability of a bad event for $k$ consecutive views as $P_b^k$. This shows that $P_b^k$ quickly approaches $0$ as $k$ increases. Therefore, eventually, an honest node will be selected as primary so that the protocol can make progress.


\section{Related Work}
\label{Section: Related Work}

Fast B4B uses block recovery mechanism that  Hermes BFT \cite{jalalzai2020hermes}  employed  to retrieve a block committed by at most $f$ nodes. During block recovery Hermes primary may collect proofs that the most recent block is not committed. Hermes is designed for optimal bandwidth usage rather than latency. It has five communication steps compare to the two in Fast B4B. The concept of collecting proofs for uncommitted block during a view change is also used in a recent concurrent manuscript \cite{No-Commit-Proofs} (called No-Commit Proofs). But \cite{No-Commit-Proofs} achieves consensus in at least four communication steps.

Kursawe \cite{OptimisticBFT3f} in his paper proposes an optimistic fast BFT consensus protocol that can achieve consensus in two-step communication only if all the nodes $n=3f+1$ are honest. Otherwise, the protocol will switch to randomized asynchronous consensus.

Martin and Alvisi \cite{FastBFT} presented FaB Paxos (Fast Byzantine Tolerant Paxos) that can tolerate $f < \frac{n}{5}$ Byzantine nodes. They also present a paramaterized version of FaB Paxos in which $n=3f+2t+1$, (where $t \leq f$) that achieve consensus in two communication steps while tolerating only $t$ Byzantine nodes. They also claim that $t=\frac{n-3f-1}{2}$, is the optimal resilience of the Fab Paxos. For $t=f$, the optimal resiliency is $f=\frac{n-1}{5}$.

Zyzzyva \cite{Kotla:2008:ZSB:1400214.1400236} is another protocol that can execute a proposal optimistically if all nodes follow the protocol. If not, then the protocol may fall back to three communication steps or may revoke the executed proposal. The main problem with Zyzzyva is that it depends (trusts) on the client for message propagation. Depending on a client to achieve consensus is not a safe option for blockchain. 

hBFT is another protocol that claims that it can achieve consensus in two communication steps while $f = \frac{n-1}{3}$. But authors in \cite{shrestha2019revisiting} show that hBFT cannot guarantee safety during view change. 

Abraham et al. \cite{GoodCase-Latency-PODC} considers the problem of reliable broadcast. Then they show that it is possible to achieve consensus in two communication steps with upper bound on Byzantine faults is $f = \frac{n+1}{5}$ and the primary is honest. Their result applies to both partially synchronous and synchronous models.

Kuznetsov et. al \cite{Revisiting-Optimal-Resilience-of-Fast-Byzantine-Consensus} generalizes the result of optimal resiliency to achieve consensus in two rounds with upper bound on faults is $f = \frac{n+1}{5}$. There result shows that this optimal resiliency bound can be achieved for protocols that employ primary to reach consensus along with those that do not.

Whereas, our protocol Fast B4B achieves consensus in two communication steps (when the primary is honest) while improving the upper bound on tolerating Byzantine faults to $f = \frac{n-1}{3}$.
\section{Conclusion}
\label{Section: Conclusion}
       In this paper, we present Fast B4B consensus protocol, that achieves consensus during normal protocol operation in just two communication steps. We show that the previous optimal bound for Byzantine resilience fast BFT can be improved from $f \leq \frac{n+1}{5}$ to $f \leq \frac{n-1}{3}$. Furthermore, Fast B4B does not employ any trusted hardware to achieve this improvement.


\bibliographystyle{IEEEtran}
\bibliography{new}

\begin{thebibliography}{10}
\providecommand{\url}[1]{#1}
\csname url@samestyle\endcsname
\providecommand{\newblock}{\relax}
\providecommand{\bibinfo}[2]{#2}
\providecommand{\BIBentrySTDinterwordspacing}{\spaceskip=0pt\relax}
\providecommand{\BIBentryALTinterwordstretchfactor}{4}
\providecommand{\BIBentryALTinterwordspacing}{\spaceskip=\fontdimen2\font plus
\BIBentryALTinterwordstretchfactor\fontdimen3\font minus
  \fontdimen4\font\relax}
\providecommand{\BIBforeignlanguage}[2]{{%
\expandafter\ifx\csname l@#1\endcsname\relax
\typeout{** WARNING: IEEEtran.bst: No hyphenation pattern has been}%
\typeout{** loaded for the language `#1'. Using the pattern for}%
\typeout{** the default language instead.}%
\else
\language=\csname l@#1\endcsname
\fi
#2}}
\providecommand{\BIBdecl}{\relax}
\BIBdecl

\bibitem{Consensus-First-Paper}
\BIBentryALTinterwordspacing
M.~Pease, R.~Shostak, and L.~Lamport, ``Reaching agreement in the presence of
  faults,'' \emph{J. ACM}, vol.~27, no.~2, p. 228–234, Apr. 1980. [Online].
  Available: \url{https://doi.org/10.1145/322186.322188}
\BIBentrySTDinterwordspacing

\bibitem{lamport1998thePaxos}
\BIBentryALTinterwordspacing
L.~Lamport, ``The part-time parliament,'' \emph{ACM Transactions on Computer
  Systems 16, 2 (May 1998), 133-169. Also appeared as SRC Research Report 49.
  This paper was first submitted in 1990, setting a personal record for
  publication delay that has since been broken by [60].}, May 1998, aCM SIGOPS
  Hall of Fame Award in 2012. [Online]. Available:
  \url{https://www.microsoft.com/en-us/research/publication/part-time-parliament/}
\BIBentrySTDinterwordspacing

\bibitem{lamport2005generalizedPaxos}
\BIBentryALTinterwordspacing
------, ``Generalized consensus and paxos,'' Tech. Rep. MSR-TR-2005-33, March
  2005. [Online]. Available:
  \url{https://www.microsoft.com/en-us/research/publication/generalized-consensus-and-paxos/}
\BIBentrySTDinterwordspacing

\bibitem{EPaxos}
\BIBentryALTinterwordspacing
I.~Moraru, D.~G. Andersen, and M.~Kaminsky, ``There is more consensus in
  egalitarian parliaments,'' in \emph{Proceedings of the Twenty-Fourth ACM
  Symposium on Operating Systems Principles}, ser. SOSP '13.\hskip 1em plus
  0.5em minus 0.4em\relax New York, NY, USA: Association for Computing
  Machinery, 2013, p. 358–372. [Online]. Available:
  \url{https://doi.org/10.1145/2517349.2517350}
\BIBentrySTDinterwordspacing

\bibitem{EpaxosRevisted}
\BIBentryALTinterwordspacing
S.~Tollman, S.~J. Park, and J.~Ousterhout, ``Epaxos revisited,'' in \emph{18th
  {USENIX} Symposium on Networked Systems Design and Implementation ({NSDI}
  21)}.\hskip 1em plus 0.5em minus 0.4em\relax {USENIX} Association, Apr. 2021,
  pp. 613--632. [Online]. Available:
  \url{https://www.usenix.org/conference/nsdi21/presentation/tollman}
\BIBentrySTDinterwordspacing

\bibitem{Castro:1999:PBF:296806.296824}
M.~Castro and B.~Liskov, ``Practical {Byzantine} fault tolerance,'' in
  \emph{Proceedings of the Third Symposium on Operating Systems Design and
  Implementation}, ser. OSDI '99.\hskip 1em plus 0.5em minus 0.4em\relax
  Berkeley, CA, USA: USENIX Association, 1999, pp. 173--186.

\bibitem{SBFT}
G.~Golan{-}Gueta, I.~Abraham, S.~Grossman, D.~Malkhi, B.~Pinkas, M.~K. Reiter,
  D.~Seredinschi, O.~Tamir, and A.~Tomescu, ``{SBFT:} a scalable decentralized
  trust infrastructure for blockchains,'' \emph{CoRR}, vol. abs/1804.01626,
  2018.

\bibitem{Jalal-Window}
M.~M. {Jalalzai} and C.~{Busch}, ``Window based {BFT} blockchain consensus,''
  in \emph{iThings, IEEE GreenCom, IEEE (CPSCom) and IEEE SSmartData 2018},
  July 2018, pp. 971--979.

\bibitem{Proteus1}
M.~M. {Jalalzai}, C.~{Busch}, and G.~G. {Richard}, ``Proteus: A scalable bft
  consensus protocol for blockchains,'' in \emph{2019 IEEE International
  Conference on Blockchain (Blockchain)}, 2019, pp. 308--313.

\bibitem{HotStuff}
M.~Yin, D.~Malkhi, M.~K. Reiter, G.~G. Gueta, and I.~Abraham, ``Hotstuff: Bft
  consensus with linearity and responsiveness,'' in \emph{Proceedings of the
  2019 ACM PODC}, ser. PODC '19.\hskip 1em plus 0.5em minus 0.4em\relax New
  York, NY, USA: Association for Computing Machinery, 2019, p. 347–356.

\bibitem{ByzantineTolerantGeneralizedPaxos}
\BIBentryALTinterwordspacing
M.~Pires, S.~Ravi, and R.~Rodrigues, ``Generalized paxos made byzantine (and
  less complex),'' \emph{Algorithms}, vol.~11, no.~9, 2018. [Online].
  Available: \url{https://www.mdpi.com/1999-4893/11/9/141}
\BIBentrySTDinterwordspacing

\bibitem{Aardvark}
A.~Clement, E.~Wong, L.~Alvisi, M.~Dahlin, and M.~Marchetti, ``Making byzantine
  fault tolerant systems tolerate byzantine faults,'' in \emph{Proceedings of
  the 6th USENIX Symposium on Networked Systems Design and Implementation},
  ser. NSDI'09.\hskip 1em plus 0.5em minus 0.4em\relax USA: USENIX Association,
  2009, p. 153–168.

\bibitem{Prime}
\BIBentryALTinterwordspacing
Y.~Amir, B.~Coan, J.~Kirsch, and J.~Lane, ``Byzantine replication under
  attack,'' in \emph{2008 IEEE International Conference on Dependable Systems
  and Networks With FTCS and DCC (DSN)}, vol.~00, June 2008, pp. 197--206.
  [Online]. Available:
  \url{doi.ieeecomputersociety.org/10.1109/DSN.2008.4630088}
\BIBentrySTDinterwordspacing

\bibitem{avarikioti2021fnfbft}
Z.~Avarikioti, L.~Heimbach, R.~Schmid, L.~Vanbever, R.~Wattenhofer, and
  P.~Wintermeyer, ``Fnf-bft: Exploring performance limits of bft protocols,''
  2021.

\bibitem{Zeno}
A.~Singh, P.~Fonseca, P.~Kuznetsov, R.~Rodrigues, and P.~Maniatis, ``Zeno:
  Eventually consistent byzantine-fault tolerance,'' in \emph{Proceedings of
  the 6th USENIX Symposium on Networked Systems Design and Implementation},
  ser. NSDI'09.\hskip 1em plus 0.5em minus 0.4em\relax USA: USENIX Association,
  2009, p. 169–184.

\bibitem{Atlas}
\BIBentryALTinterwordspacing
V.~Enes, C.~Baquero, T.~F. Rezende, A.~Gotsman, M.~Perrin, and P.~Sutra,
  ``State-machine replication for planet-scale systems,'' in \emph{Proceedings
  of the Fifteenth European Conference on Computer Systems}, ser. EuroSys
  '20.\hskip 1em plus 0.5em minus 0.4em\relax New York, NY, USA: Association
  for Computing Machinery, 2020. [Online]. Available:
  \url{https://doi.org/10.1145/3342195.3387543}
\BIBentrySTDinterwordspacing

\bibitem{Fault-Tolerance-System}
\BIBentryALTinterwordspacing
T.~Distler, ``Byzantine fault-tolerant state-machine replication from a systems
  perspective,'' \emph{ACM Comput. Surv.}, vol.~54, no.~1, Feb. 2021. [Online].
  Available: \url{https://doi.org/10.1145/3436728}
\BIBentrySTDinterwordspacing

\bibitem{Upright-Services}
\BIBentryALTinterwordspacing
A.~Clement, M.~Kapritsos, S.~Lee, Y.~Wang, L.~Alvisi, M.~Dahlin, and T.~Riche,
  ``Upright cluster services,'' in \emph{Proceedings of the ACM SIGOPS 22nd
  Symposium on Operating Systems Principles}, ser. SOSP '09.\hskip 1em plus
  0.5em minus 0.4em\relax New York, NY, USA: Association for Computing
  Machinery, 2009, p. 277–290. [Online]. Available:
  \url{https://doi.org/10.1145/1629575.1629602}
\BIBentrySTDinterwordspacing

\bibitem{BFT-Critical-Infrastructure}
A.~N. Bessani, P.~Sousa, M.~Correia, N.~F. Neves, and P.~Veríssimo, ``The
  crutial way of critical infrastructure protection,'' \emph{IEEE Security
  Privacy}, vol.~6, no.~6, pp. 44--51, 2008.

\bibitem{BFT-Scada}
A.~Nogueira, M.~Garcia, A.~Bessani, and N.~Neves, ``On the challenges of
  building a bft scada,'' in \emph{2018 48th Annual IEEE/IFIP International
  Conference on Dependable Systems and Networks (DSN)}, 2018, pp. 163--170.

\bibitem{BFT-DataBases}
\BIBentryALTinterwordspacing
R.~Garcia, R.~Rodrigues, and N.~Pregui\c{c}a, ``Efficient middleware for
  byzantine fault tolerant database replication,'' in \emph{Proceedings of the
  Sixth Conference on Computer Systems}, ser. EuroSys '11.\hskip 1em plus 0.5em
  minus 0.4em\relax New York, NY, USA: Association for Computing Machinery,
  2011, p. 107–122. [Online]. Available:
  \url{https://doi.org/10.1145/1966445.1966456}
\BIBentrySTDinterwordspacing

\bibitem{Monoxide}
\BIBentryALTinterwordspacing
J.~Wang and H.~Wang, ``Monoxide: Scale out blockchains with asynchronous
  consensus zones,'' in \emph{16th {USENIX} Symposium on Networked Systems
  Design and Implementation ({NSDI} 19)}.\hskip 1em plus 0.5em minus
  0.4em\relax Boston, MA: {USENIX} Association, Feb. 2019, pp. 95--112.
  [Online]. Available:
  \url{https://www.usenix.org/conference/nsdi19/presentation/wang-jiaping}
\BIBentrySTDinterwordspacing

\bibitem{jalalzai2021fasthotstuff}
M.~M. Jalalzai, J.~Niu, C.~Feng, and F.~Gai, ``Fast-hotstuff:a fast and
  resilient hotstuff protocol,'' 2021.

\bibitem{Algorand}
Y.~Gilad, R.~Hemo, S.~Micali, G.~Vlachos, and N.~Zeldovich, ``Algorand: Scaling
  {Byzantine} agreements for cryptocurrencies,'' in \emph{Proceedings of the
  26th Symposium on Operating Systems Principles}, ser. SOSP '17.\hskip 1em
  plus 0.5em minus 0.4em\relax New York, NY, USA: ACM, 2017, pp. 51--68.

\bibitem{Casper}
V.~Buterin and V.~Griffith, ``Casper the friendly finality gadget,''
  \emph{arXiv preprint arXiv:1710.09437}, 2017.

\bibitem{BFT-SMART}
A.~{Bessani}, J.~{Sousa}, and E.~E.~P. {Alchieri}, ``State machine replication
  for the masses with bft-smart,'' in \emph{2014 44th Annual IEEE/IFIP
  International Conference on Dependable Systems and Networks}, 2014, pp.
  355--362.

\bibitem{jalalzai2020hermes}
M.~Jalalzai, C.~Feng, C.~Busch, G.~R. III, and J.~Niu, ``The hermes bft for
  blockchains,'' \emph{IEEE Transactions on Dependable and Secure Computing},
  no.~01, pp. 1--1, sep 5555.

\bibitem{OptimisticBFT3f}
K.~Kursawe, ``Optimistic byzantine agreement,'' in \emph{21st IEEE Symposium on
  Reliable Distributed Systems, 2002. Proceedings.}, 2002, pp. 262--267.

\bibitem{FastBFT}
\BIBentryALTinterwordspacing
J.-P. Martin and L.~Alvisi, ``Fast byzantine consensus,'' \emph{IEEE Trans.
  Dependable Secur. Comput.}, vol.~3, no.~3, p. 202–215, Jul. 2006. [Online].
  Available: \url{https://doi.org/10.1109/TDSC.2006.35}
\BIBentrySTDinterwordspacing

\bibitem{Revisiting-Optimal-Resilience-of-Fast-Byzantine-Consensus}
\BIBentryALTinterwordspacing
P.~Kuznetsov, A.~Tonkikh, and Y.~X. Zhang, ``Revisiting optimal resilience of
  fast byzantine consensus,'' in \emph{Proceedings of the 2021 ACM Symposium on
  Principles of Distributed Computing}, ser. PODC'21.\hskip 1em plus 0.5em
  minus 0.4em\relax New York, NY, USA: Association for Computing Machinery,
  2021, p. 343–353. [Online]. Available:
  \url{https://doi.org/10.1145/3465084.3467924}
\BIBentrySTDinterwordspacing

\bibitem{GoodCase-Latency-PODC}
\BIBentryALTinterwordspacing
I.~Abraham, K.~Nayak, L.~Ren, and Z.~Xiang, ``Good-case latency of byzantine
  broadcast: a complete categorization,'' in \emph{{PODC} '21: {ACM} Symposium
  on Principles of Distributed Computing, Virtual Event, Italy, July 26-30,
  2021}, A.~Miller, K.~Censor{-}Hillel, and J.~H. Korhonen, Eds.\hskip 1em plus
  0.5em minus 0.4em\relax {ACM}, 2021, pp. 331--341. [Online]. Available:
  \url{https://doi.org/10.1145/3465084.3467899}
\BIBentrySTDinterwordspacing

\bibitem{MinBFT}
G.~S. Veronese, M.~Correia, A.~N. Bessani, L.~C. Lung, and P.~Verissimo,
  ``Efficient byzantine fault-tolerance,'' \emph{IEEE Transactions on
  Computers}, vol.~62, no.~1, pp. 16--30, 2013.

\bibitem{TEE-Vulnerability}
D.~Cerdeira, N.~Santos, P.~Fonseca, and S.~Pinto, ``Sok: Understanding the
  prevailing security vulnerabilities in trustzone-assisted tee systems,'' in
  \emph{2020 IEEE Symposium on Security and Privacy (SP)}, 2020, pp.
  1416--1432.

\bibitem{Fischer:1985:IDC:3149.214121}
M.~J. Fischer, N.~A. Lynch, and M.~S. Paterson, ``Impossibility of distributed
  consensus with one faulty process,'' \emph{J. ACM}, vol.~32, no.~2, pp.
  374--382, Apr. 1985.

\bibitem{Dwork:1988:CPP:42282.42283}
C.~Dwork, N.~Lynch, and L.~Stockmeyer, ``Consensus in the presence of partial
  synchrony,'' \emph{J. ACM}, vol.~35, no.~2, pp. 288--323, Apr. 1988.

\bibitem{Short-Signatures-from-the-Weil-Pairing}
D.~Boneh, B.~Lynn, and H.~Shacham, ``Short signatures from the weil pairing,''
  in \emph{Proceedings of the 7th International Conference on the Theory and
  Application of Cryptology and Information Security: Advances in Cryptology},
  ser. ASIACRYPT ’01.\hskip 1em plus 0.5em minus 0.4em\relax Berlin,
  Heidelberg: Springer-Verlag, 2001, p. 514–532.

\bibitem{Boneh:2003}
D.~Boneh, C.~Gentry, B.~Lynn, and H.~Shacham, ``Aggregate and verifiably
  encrypted signatures from bilinear maps,'' in \emph{Proceedings of the 22nd
  International Conference on Theory and Applications of Cryptographic
  Techniques}.\hskip 1em plus 0.5em minus 0.4em\relax Berlin, Heidelberg:
  Springer-Verlag, 2003, pp. 416--432.

\bibitem{BDNSignatureScheme}
D.~Boneh, M.~Drijvers, and G.~Neven, ``Compact multi-signatures for smaller
  blockchains,'' in \emph{Advances in Cryptology -- ASIACRYPT 2018}, T.~Peyrin
  and S.~Galbraith, Eds.\hskip 1em plus 0.5em minus 0.4em\relax Cham: Springer
  International Publishing, 2018, pp. 435--464.

\bibitem{Lamport:1982:BGP:357172.357176}
L.~Lamport, R.~Shostak, and M.~Pease, ``The {Byzantine} generals problem,''
  \emph{ACM Trans. Program. Lang. Syst.}, vol.~4, no.~3, pp. 382--401, Jul.
  1982.

\bibitem{Coin-Flipping-by-Telephone}
\BIBentryALTinterwordspacing
M.~Blum, ``Coin flipping by telephone a protocol for solving impossible
  problems,'' \emph{SIGACT News}, vol.~15, no.~1, p. 23–27, Jan. 1983.
  [Online]. Available: \url{https://doi.org/10.1145/1008908.1008911}
\BIBentrySTDinterwordspacing

\bibitem{Coin-Flipping-First-Paper}
M.~Ben-Or and N.~Linial, ``Collective coin flipping, robust voting schemes and
  minima of banzhaf values,'' in \emph{26th Annual Symposium on Foundations of
  Computer Science (sfcs 1985)}, 1985, pp. 408--416.

\bibitem{Coin-Flipping-Popov}
\BIBentryALTinterwordspacing
S.~Popov, ``On a decentralized trustless pseudo-random number generation
  algorithm,'' \emph{Journal of Mathematical Cryptology}, vol.~11, no.~1, pp.
  37--43, 2017. [Online]. Available:
  \url{https://doi.org/10.1515/jmc-2016-0019}
\BIBentrySTDinterwordspacing

\bibitem{Random-Oracles-in-Constantinople}
C.~Cachin, K.~Kursawe, and V.~Shoup, ``Random oracles in constantinople:
  Practical asynchronous byzantine agreement using cryptography,'' \emph{J.
  Cryptol.}, vol.~18, no.~3, p. 219–246, Jul. 2005.

\bibitem{An-Optimally-Fair-Coin-Toss}
T.~Moran, M.~Naor, and G.~Segev, ``An optimally fair coin toss,'' in
  \emph{Theory of Cryptography}, O.~Reingold, Ed.\hskip 1em plus 0.5em minus
  0.4em\relax Berlin, Heidelberg: Springer Berlin Heidelberg, 2009, pp. 1--18.

\bibitem{hanke2018dfinity}
T.~Hanke, M.~Movahedi, and D.~Williams, ``Dfinity technology overview series,
  consensus system,'' 2018.

\bibitem{No-Commit-Proofs}
N.~Giridharan, H.~Howard, I.~Abraham, N.~Crooks, and A.~Tomescu, ``No-commit
  proofs: Defeating livelock in bft,'' Cryptology ePrint Archive, Report
  2021/1308, 2021, \url{https://ia.cr/2021/1308}.

\bibitem{Kotla:2008:ZSB:1400214.1400236}
R.~Kotla, A.~Clement, E.~Wong, L.~Alvisi, and M.~Dahlin, ``Zyzzyva: Speculative
  {Byzantine} fault tolerance,'' \emph{Commun. ACM}, vol.~51, no.~11, pp.
  86--95, Nov. 2008.

\bibitem{shrestha2019revisiting}
N.~Shrestha, M.~Kumar, and S.~Duan, ``Revisiting hbft: Speculative byzantine
  fault tolerance with minimum cost,'' 2019.

\end{thebibliography}

\end{document}